\documentclass[conference]{IEEEtran}
\IEEEoverridecommandlockouts

\usepackage{balance}
\usepackage{comment}
\usepackage{mathtools}
\usepackage{amsthm}
\usepackage{amsfonts}
\usepackage{braket}
\usepackage{physics}
\usepackage{braket}
\usepackage{xspace}
\usepackage{scalerel}
\usepackage{tikz}
\usetikzlibrary{svg.path}

\definecolor{orcidlogocol}{HTML}{A6CE39}
\tikzset{
  orcidlogo/.pic={
    \fill[orcidlogocol] svg{M256,128c0,70.7-57.3,128-128,128C57.3,256,0,198.7,0,128C0,57.3,57.3,0,128,0C198.7,0,256,57.3,256,128z};
    \fill[white] svg{M86.3,186.2H70.9V79.1h15.4v48.4V186.2z}
                 svg{M108.9,79.1h41.6c39.6,0,57,28.3,57,53.6c0,27.5-21.5,53.6-56.8,53.6h-41.8V79.1z M124.3,172.4h24.5c34.9,0,42.9-26.5,42.9-39.7c0-21.5-13.7-39.7-43.7-39.7h-23.7V172.4z}
                 svg{M88.7,56.8c0,5.5-4.5,10.1-10.1,10.1c-5.6,0-10.1-4.6-10.1-10.1c0-5.6,4.5-10.1,10.1-10.1C84.2,46.7,88.7,51.3,88.7,56.8z};
  }
}

\newcommand\orcidicon[1]{\href{https://orcid.org/#1}{\mbox{\scalerel*{
\begin{tikzpicture}[yscale=-1,transform shape]
\pic{orcidlogo};
\end{tikzpicture}
}{|}}}}

\usepackage{hyperref}
\hypersetup{
  colorlinks=true,
  linkcolor=blue!50!blue,
  urlcolor=blue!70!black,
  citecolor=black,
}

\usepackage{cite}
\usepackage{soul}
\usepackage{comment}
\usepackage{amsmath,amssymb,amsfonts}
\usepackage{algorithmic}
\usepackage{graphicx}
\usepackage{url}
\usepackage{textcomp}
\usepackage{xcolor}
\usepackage{CJKutf8}
\def\BibTeX{{\rm B\kern-.05em{\sc i\kern-.025em b}\kern-.08em
    T\kern-.1667em\lower.7ex\hbox{E}\kern-.125emX}}

\ifCLASSOPTIONcompsoc
\usepackage[caption=false,font=normalsize,labelfont=sf,textfont=sf]{subfig}
\else
\usepackage[caption=false,font=footnotesize]{subfig}
\fi

\newtheorem{definition}{Definition}
\newtheorem{theorem}{Theorem}

\newtheorem{lemma}{Lemma}
\newtheorem{example}{Example}%

\usepackage{glossaries}
\usepackage{glossaries}
\newacronym{ai}{AI}{Artificial Intelligence}
\newacronym{ml}{ML}{Machine Learning}
\newacronym{qml}{QML}{Quantum ML}
\newacronym{xqml}{XQML}{eXplainable Quantum ML}
\newacronym{nserc}{NSERC}{Natural Sciences and Engineering Research Council}

\begin{document}

\title{Quantum Algorithms for Shapley Value Calculation}

\author{\IEEEauthorblockN{Iain Burge\IEEEauthorrefmark{1}~\orcidicon{0009-0008-5360-9430},
Michel Barbeau\IEEEauthorrefmark{1}~\orcidicon{0000-0003-3531-4926},
Joaquin Garcia-Alfaro\IEEEauthorrefmark{2}~\orcidicon{0000-0002-7453-4393}}
\IEEEauthorblockA{\IEEEauthorrefmark{1}School of Computer Science, Carleton University, Ottawa, Canada}
\IEEEauthorblockA{\IEEEauthorrefmark{2}SAMOVAR, T\'el\'ecom SudParis, Institut Polytechnique de Paris, Palaiseau, France}}

\maketitle

\begin{abstract}
In the classical context, the cooperative game theory concept of the Shapley value has been adapted for post hoc explanations of \gls*{ml} models. 
However, this approach does not easily translate to eXplainable Quantum ML (XQML). 
Finding Shapley values can be highly computationally complex. 
We propose quantum algorithms which can extract Shapley values within some confidence interval. 
Our results perform in polynomial time. 
We demonstrate the validity of each approach under specific examples of cooperative voting games.
\end{abstract}

\begin{IEEEkeywords}
Shapley Value, Quantum Computing, Cooperative Game Theory, Explainable Quantum Machine Learning, Machine Learning, Artificial Intelligence, Quantum Machine Learning.
\end{IEEEkeywords}

\section{Introduction}

Research in \gls*{ml} over the past decades deals with black box models. 
Unfortunately, black box models are inherently difficult to interpret. 
Inherently interpretable models would likely be best~\cite{rudin2019stop}, as an explanation of an interpretable model is guaranteed to be correct. 
However, we ideally do not want to discard all of the previous research using black box models.
As a result, there has been a substantial effort in implementing and improving post hoc explanation methods.
Similarly, \gls*{xqml} may benefit from adapting post hoc explanation methods to the quantum realm.

One of the most popular methods of generating post hoc explanations involves calculating Shapley values.
Yet, classical strategies to approximate Shapley values are unwieldy to apply in the context of quantum computers.
Therefore, it is necessary to explore a more native quantum solution to Shapley value approximation.

In this paper, we develop a flexible framework for the global evaluation of input factors in quantum circuits that approximates the Shapley values of such factors. 
Our framework has a one time increased circuit complexity of an additional $\mathcal{O}(an \log an)$ c-not gates, with a total additive increase in circuit depth of $\mathcal{O}(an)$, where $n$ is the number of factors, and $a>0$ is a real number.
The change in space complexity for global evaluations is an additional $\mathcal{O}(\log an)$ qubits over the evaluated circuit.
The circuit of increased complexity must then be repeated $\mathcal{O}(\epsilon^{-1})$ times. 
This procedure can achieve an error of $\mathcal{O}(a^{-1}+\epsilon)$ multiplied by a problem dependent upper bound.
This starkly contrasts the $\mathcal{O}(2^n)$ assessments needed to calculate the Shapley values under the general case directly.
It is also better than the $\mathcal{O}(\sigma^2 \epsilon^{-2})$ complexity given by Monte Carlo approaches, where $\sigma$ is standard deviation \cite{montanaro2015quantum}. 

\medskip

The paper is organized as follows\footnote{Peer-reviewed version~\cite{QCE2023}.}.
Section~\ref{section:relwork} surveys related work.
Section~\ref{section:background} provides background and preliminaries.
Section~\ref{sec:WeightedVotingGame} introduces a guiding example problem.
Sections~\ref{sec:QuantumRepresentation} and~\ref{sec:algorithm} present our methods. 
Section~\ref{sec:example} demonstrates the application of our methods to the problem discussed in Section~\ref{sec:WeightedVotingGame}.
Section~\ref{sec:conclusion} concludes the work.

\section{Related Work}
\label{section:relwork}

Shapley values have been widely used to address multiple engineering problems, including regression, statistical analysis, and machine learning~\cite{lipovetsky2023quantum}.
Finding Shapley values presents a difficult computational combinatorial problem. 
Our work proposes a novel quantum algorithm that reduces this combinatorial problem to an estimation problem which can leverage the power of quantum computation.
Our approach performs in polynomial time. 
We apply our method to weighted voting games~\cite{matsui2001np}.

The deterministic computation of Shapley values in the context of weighted voting games is as difficult as NP-Hard \cite{matsui2001np, prasad1990np}. 
Since voting games are some of the simplest cooperative games, this result does not bode well for more complex scenarios. In the context of Shapley values for machine learning, it has also been shown that the calculation of Shapley values are not tractable for even regression models \cite{van2022tractability}. It was also proven that on the empirical distribution, finding a Shapley value takes exponential time \cite{bertossi2020causality}.

The literature has also addressed the use of Shapley values on \gls*{qml}. Indeed, \gls*{xqml} aims at adding explainability behind model predictions, e.g., in addition to providing classification~\cite{mercaldo2022towards}. \gls*{xqml} can be considered as an alternative research direction of \gls*{qml} instead of trying to justify quantum advantage~\cite{Schuld_2022}. The eventual goal of \gls*{xqml} is to provide, in addition to predictions, humanly understandable interpretations of the predictive models, e.g., for malware detection and classification, under a cybersecurity context~\cite{steinmuller2022explainable}.

Recent work by Heese et al. extends the notion of feature importance for model
predictions in classical \gls*{ml}, to the \gls*{qml} realm~\cite{heese2023explaining}. In contrast to classical ML methods, in which Shapley values are applied to evaluate the importance of each feature for model predictions, Heese et al. apply Shapley values to evaluate the relevance of each  quantum gate associated with a given parametrized quantum circuit. In their work, Heese et al. compute Shapley values involving stochastic processing. In general, the algorithms presented in this paper improve stochastic computation of Shapley values. 
Moreover, our work justifies the interest in handling post hoc explanation frameworks such as the one of Heese et al., in a quantum extended manner.

\section{Shapley Values}
\label{section:background}

Cooperative game theory is the study of coalitional games. 
In this article, we are most interested in Shapley values. 
We now list some definitions and preliminaries,

\begin{definition}[Coalitional game]
    A \emph{coalitional game} is described as the pair $G=(F, V)$. 
    $F = \{0, 1, ..., n\}$ is a set of $n+1$ players.
    $V: \mathcal{P} (F) \xrightarrow{} \mathbb{R}$ is a value function with $V(S) \in \mathbb{R}$ representing the value of a given coalition $S \subseteq F$, with the restriction that $V(\emptyset) = 0$.
\end{definition}
\begin{definition}[Payoff vector]
    Given a game $G=(F, V)$, there exists a \emph{payoff vector} $\Phi(G)$ of length $n+1$. 
    Each element $\Phi(G)_i\in \mathbb{R}$ represents the utility of player $i\in F$. 
    A payoff vector is determined by the value function.
    Player $i$'s payoff value $\Phi(G)_i$ is determined by how $V(S)$, $S\subseteq F$, is affected by $i$'s inclusion or exclusion from $S$.
\end{definition}

There are a variety of solution concepts for constructing payoffs~\cite{aumann2010some}.
We focus on the Shapley solution concept, which returns a payoff vector~\cite{winter2002shapley}.
Each element of the payoff vector $\Phi(G)_i$, is called player $i$'s Shapley value.
Shapley values have multiple different interpretations depending on the game being analyzed.

Shapley values are derived using one of several sets of axioms. 
We use the following four \cite{winter2002shapley}.
Suppose we have games $G=(F, V)$ and $G'=(F, V')$, and a payoff vector $\Phi(G)$, then:
\begin{enumerate}
    \item Efficiency: The sum of all utility is equal to the utility of the grand coalition (the coalition containing all players),
        \begin{equation*}
            \sum\limits_{i=1}^{n} \Phi(G)_i = V(F)
        \end{equation*}
        
    \item Equal Treatment: Players $i$, $j$ are said to be symmetrical if for all $S\subseteq F$, where $i,j\notin S$ we have that $V(S\cup \{i\}) = V(S \cup \{j\})$.
    If $i$, and $j$ are symmetric in $G$, then they are treated equally, $\Phi(G)_i = \Phi(G)_j$.
    
    \item Null Player: Consider a player $i \in F$, if for all $S\subseteq F$ such that $i\notin S$, we have $V(S) = V(S\cup \{i\})$, then $i$ is a null player.
    If $i$ is a null player then $\Phi(G)_i = 0$.
    
    \item Additivity: If a player is in two games $G$ and $G'$, then the Shapley values of the two games is additive
        \begin{align*}
            \Phi(G+G')_i = \Phi(G)_i + \Phi(G')_i
        \end{align*}
    where a game $G+G'$ is defined as $(F, V+V')$, and $(V+V')(S) = V(S) + V'(S)$, $S \subseteq F$.
\end{enumerate}
These axioms lead to a single unique and  intuitive division of utility~\cite{winter2002shapley}.
The values of the payoff vectors can be interpreted as the responsibility of the respective players for the final outcome~\cite{hart1989shapley}.
When player $i$ has a small payoff $\Phi(G)_i$, then player $i$ has a neutral impact on the final outcome.
When player $i$ has a large payoff, then player $i$ has a large impact on the final outcome.

The Shapley value of $i$ turns out to be the expected marginal contribution to a random coalition $S \subseteq F \setminus \{i\}$, 
where the marginal contribution is equal to $V(S \cup \{i\})- V(S)$ \cite{hart1989shapley}. 

\begin{definition}[Shapley value \cite{shapley1952value}]\label{def:shapleyvalue}
    Let $G=(F, V)$, for simplicity sake, we write $\Phi(G)_i$ as $\Phi_i$.
    The Shapley value of the $i^{th}$ player $\Phi_i$ is described as:
    \begin{equation}\label{eq:payoff}
        \Phi_i = \sum\limits_{S \subseteq F \setminus \{i\}} \gamma(\lvert F \setminus \{i\} \rvert, \lvert S \rvert) \cdot (V(S\cup \{i\}) - V(S))
    \end{equation}
    where $n=\lvert F \setminus \{i\} \rvert$, and
    \begin{equation*}
        \gamma(n, m) = \frac{1}{{n \choose m} (n+1)} = \frac{m!(n-m)!}{(n+1)!}
    \end{equation*}
\end{definition}

For this paper, we will define algorithms for a special case of games, monotonic games (Definition~\ref{def:monotonic}).
\begin{definition}[monotonic game \cite{shapley1952value}]
\label{def:monotonic}
    A game is \emph{monotonic} if for all $S,H \subseteq F$, we have, $V(S\cup H) \geq V(S)$.
\end{definition}
Note that when a game is monotonic, every summand in Equation~\eqref{eq:payoff} is non-negative.

\section{Weighted Voting Games}
\label{sec:WeightedVotingGame}
Let us consider a scenario where Shapley values correspond to voting power.
The voting power of a player represents how many instances for which that player has the deciding vote.
Three friends sit around a table.
They are deliberating a grave matter.
Should they get Chinese food for the second weekend in a row?
They decide they should take a vote.
Alice just got a promotion at work.
To celebrate this, her friends agreed to give her three votes.
Bob, who is the youngest of the group, also had good news, an incredible mark on his latest assignment!
So, everyone decided Bob should get two votes.
Charley, who had nothing to celebrate, gets one vote.
The group decides to go out for Chinese food if there are four \emph{yes} votes.
In the end, all of the friends vote for Chinese food.

At the restaurant, they run into their friend David, who is a mathematician.
David is intrigued when he hears about their vote. He begins to wonder how much power each friend had in the vote.
The intuitive answer is that Alice had the most voting power, Bob had the second most, and Charley had the least.
However, David notices something strange. 
Bob does not seem to have a more meaningful influence than Charley.
There are no cases where Bob's two votes would do more than Charley's single vote.
So, David concludes, there must be a more nuanced answer.
Lucky for David, he has heard of cooperative game theory and Shapley values.
So, he might be able to answer his question.

How can Definition~\ref{def:shapleyvalue} be applied to David's problem?
Let us define the game $G=(F,V)$.
The players are Alice (0), Bob (1), and Charley (2) represented as $F=\{0,1,2\}$.
Denote each player's voting weight as $w_0=3$, $w_1=2$, and $w_2=1$.
Recall that the vote threshold was $q=4$.
Then, for any subset $S\subseteq F$, we can define $V$ as:
\begin{equation} \label{eq:voteGameV}
     V(S) =
    \begin{cases}
        1 &\quad \text{if } \sum\limits_{j \in S} w_j \geq q,\\
        0 &\quad \text{otherwise.}
    \end{cases}
\end{equation}
So, $V(S)$ is one if the sum of players' votes in $S$ reaches the threshold of four.
Otherwise, the vote fails, and so $V(S)$ is zero.

In general, this is called a weighted voting game~\cite{matsui2001np}.
One could easily add more players with arbitrary non-negative $w_j$ and $q$. 
Note that weighted voting games fall into the family of monotonic games (Definition~\ref{def:monotonic}).

In this context, the terms in the Shapley value equation have an intuitive meaning.
Take player $i$, and consider a set $S\subseteq F\setminus \{i\}$.
If $V(S\cup \{i\})-V(S)=1$, then the $i^\text{th}$ player is a deciding vote for the set of players $S$.
Otherwise, player $i$ is not a deciding vote.

Thus, for weighted voting games, player $i$'s Shapley value represents a weighted count of how many times $i$ is a deciding vote.
We can work out the Shapley values by hand:
\begin{align*}
    V(\emptyset) &= 0 &\quad V(\{ 0, 1\})  &= 1\\
    V(\{0 \})  &= 0 &\quad V(\{ 0, 2\})  &= 1\\
    V(\{1 \})  &= 0 &\quad V(\{ 1, 2\})  &= 0\\
    V(\{2 \})  &= 0 &\quad V(\{ 0, 1, 2 \})  &= 1\\
\end{align*}
From this, we have:

\begin{footnotesize}
\begin{align*}
    \Phi_0 &= \sum\limits_{S \subseteq F \setminus \{i\}} \gamma(\lvert F \setminus \{i\} \rvert, \lvert S \rvert) \cdot (V(S\cup \{i\}) - V(S)) \\
    &= \gamma(2, 0)\cdot(V(\{0 \})-V(\emptyset)) + \gamma(2,1)\cdot(V(\{ 0, 1\})-V(\{1 \}))\\
      &\quad+ \gamma(2,1)\cdot(V(\{ 0, 2\})-V(\{2 \})) + \gamma(2,2)\cdot(V(\{ 0, 1, 2\})\\
      & -V(\{1, 2\})) = \gamma(2, 0)\cdot(0-0) + \gamma(2,1)\cdot(1-0)\\
      &+ \gamma(2,1)\cdot(1-0) + \gamma(2,2)\cdot(1-0)\\
    &= 2 \cdot \gamma(2, 1) + \gamma(2,2) = 2 \cdot \frac{1!(2-1)!}{(2+1)!} + \frac{2!(2-2)!}{(2+1)!}\\
    &= 2\cdot\frac{1}{6} + \frac{1}{3} = \frac{2}{3}
\end{align*}
\end{footnotesize}

This can be repeated to get,
\begin{equation*}
    \Phi_1, \Phi_2 = \frac{1}{6}
\end{equation*}

In the case of Alice, Bob, and Charley's voting game, it is trivial to calculate their respective Shapley values.
However, what if one-hundred colleagues were choosing a venue for a party, all with different numbers of votes?
In that case, a direct calculation would take $2^{100}$ assessments of $V$!
for this more general case, we need to be clever.

\section{Naive Quantum Representation of Shapley Value Calculation}
\label{sec:QuantumRepresentation}
We represent the Shapley value calculation problem in the quantum format.
Consider an $n+1$ player monotonic game $G$ represented by the pair $(F,V)$, where $F=\{0,1,\dots,n\}$ and $V: \mathcal{P} (F) \xrightarrow{} \mathbb{R}$, with $V(\emptyset) = 0$.
Let us define the function,
\begin{equation}\label{eq:W}
W(S) = V(S\cup\{i\})-V(S), \quad S \subseteq F\setminus\{i\}.    
\end{equation}
We define $W_\text{max}$ as an upper bound of the absolute maximum change in value when adding player $i$ to a subset $S$.
\begin{equation}\label{eq:Wmax}
    W_\text{max} 
    \geq \max\limits_{S \subseteq F\setminus \{i\}}  \lvert W(S) \rvert.
\end{equation}
Let $i\in F$ be a given player. 
Consider the selection binary sequence $x=x_0x_1\dots x_{i-1} x_{i+1}\dots x_n$.
Let $S_x$ be the set of all players $j\in F$ such that $x_j=1$.
Then $S_x$ could encode any player coalition that excludes the player $i$.
Next, define, 
\begin{equation}\label{eq:What}
    \hat{W}(x) := \frac{W\left(S_x\right)}{W_\text{max}}.
\end{equation}

We define the following block diagonal matrix:
\begin{equation}\nonumber
B = \bigoplus_{k=0}^{2^n-1} 
\begin{pmatrix}
\sqrt{1-\phi(k,n)} &   \sqrt{\phi(k,n)}  \\
\sqrt{\phi(k,n)} & -\sqrt{1-\phi(k,n)}  
\end{pmatrix}
\end{equation}
where
\begin{equation}\nonumber
\phi(k,n)
=
\gamma(n,c(k)) 
\cdot
\hat{W}(k)    
\end{equation}
and $c(k)$ is the number of ones in the binary representation of $k$, over $\log n$ bits.
\begin{theorem}
The block diagonal matrix $B(n)$ is unitary.
\end{theorem}
\begin{proof}
It follows from the fact that each block
$$
\begin{pmatrix}
\sqrt{1-\phi(k,n)} &   \sqrt{\phi(k,n)}  \\
\sqrt{\phi(k,n)} & -\sqrt{1-\phi(k,n)}  
\end{pmatrix}
$$
of $B$ is unitary.
\end{proof}

We construct the quantum system:
\begin{equation}
S = B (H^{\otimes n}\otimes I) \vert 0 \rangle^{\otimes n+1}    
\end{equation}

\begin{theorem} 
The expected values of the rightmost qubits of $S$ is
$\frac{\Phi_i}{2^n \cdot W_{max}}$.
\end{theorem}
\begin{proof}
It follows from the fact that
$$
(H^{\otimes n}\otimes I) \vert 0 \rangle^{\otimes n+1}
=
\sum_{k=0}^{2^n-1} \frac{1}{\sqrt{2^n}} \ket{k}\ket{0}
$$
and the following sequence of equivalences:
\begin{equation}\nonumber
\begin{split}
\sum_{k=0}^{2^n-1} & \left( B_{k+1,k} \right)^2  \\
 & = \sum_{k=0}^{2^n-1} \phi(k,n) = \sum_{k=0}^{2^n-1} \gamma(n,c(k)) \cdot W(k) \\
 & =\sum\limits_{S \subseteq F \setminus \{i\}} \gamma(\lvert F \setminus \{i\} \rvert, \lvert S \rvert) \cdot \frac{\left( V(S\cup \{i\}) - V(S) \right)}{W_{max}}
\end{split}
\end{equation}
Hence, the probability of measuring a one in the last qubit of the quantum system $S$ is
\begin{equation*}
\frac{1}{2^n \cdot W_{max}}
\sum_{S \subseteq F \setminus \{i\}} 
\gamma(\lvert F \setminus \{i\} \rvert, \lvert S \rvert) 
\cdot \left( V(S\cup \{i\}) - V(S) \right)
\end{equation*}
which is equivalent to 
$\frac{\Phi_i}{2^n \cdot W_{max}}$.

\end{proof}
The Shapley value $\Phi_i$ can be obtained by repeatedly creating the quantum system $S$,
measuring their last qubit, taking the average, and finally multiplying by $2^n \cdot W_{max}$.
As will be briefly discussed in Section \ref{sec:algorithm}, the value can also be extracted with a more efficient strategy.
Although, this representation requires the preparation of a quantum state with an exponential number of terms ($2^n$).
In the following section, we develop a more efficient solution for Shapley value calculation of monotonic games.

\section{Quantum Algorithm for Monotonic Games}
\label{sec:algorithm}

Consider an $n+1$ player game $G$ represented by the pair $(F,V)$, where $F=\{0,1,\dots,n\}$ and $V: \mathcal{P} (F) \xrightarrow{} \mathbb{R}$, with $V(\emptyset) = 0$.
A naive approach to finding the $i^\text{th}$ player's Shapley value is through direct calculation using the Shapley Equation~\eqref{eq:payoff}, 
completing the task in $\mathcal{O}(2^n)$ assessments of $V$.
For structured games, it is occasionally possible to calculate Shapley values more efficiently.
Otherwise, the only option is random sampling \cite{castro2009polynomial}.
In each case, there are substantial trade-offs, in this section we propose a quantum algorithm which has some substantial advantages.

Suppose we have a quantum representation of the function $W(S)$, defined in Equation~\eqref{eq:W},
which comprises two registers, a player register, and a utility register.
In the player register, the computational basis states represent different subsets of players, 
where the $j^\text{th}$ qubit represents the $j^\text{th}$ player.
In this encoding, the $j^\text{th}$ qubit being $\ket{1}$ represents $j$'s inclusion in the subset, and $\ket{0}$ represents its exclusion.
Thus, every possible subset of players has a corresponding basis state.
It follows that representing every subset of players simultaneously in a quantum superposition is possible.
The one-qubit utility register encodes the output of $W$.

Computing the Shapley value method consists of the following steps:
\newcounter{algorithmicSteps}
\begin{enumerate}
    \item Represent every possible subset of players in a quantum superposition, excluding player $i$, with amplitudes corresponding to weights ($\gamma(n,m)$) in the Shapley Equation~\eqref{eq:payoff}.
    \item Perform the quantum version of $W$ controlled by the player register, and any auxiliary registers, on the utility register.
    \setcounter{algorithmicSteps}{\value{enumi}}
\end{enumerate}
From this point, it is possible to approximate the Shapley value $\Phi_i$ if one has the expected value of measuring the utility register.
\begin{enumerate}
    \setcounter{enumi}{\value{algorithmicSteps}}
    \item To approximate the expected value of measuring the utility register, one can either:
    \begin{itemize}
        \item Use $\mathcal{O}(\epsilon^{-2})$ repetitions of steps 1 and 2 followed by a measurement of the utility register.
        \item Perform amplitude estimation with $\mathcal{O}(\epsilon^{-1})$ iterations, where step 1 is $\mathcal{A}$ and step 2 is $W$ as defined in \cite{montanaro2015quantum}.
        In opposition to the previous choice, this process increases circuit depth.
    \end{itemize}
\end{enumerate}
In either case, the error in approximating the expected value for measuring the utility register is within $\pm\epsilon$ with some predetermined likelihood, where $\epsilon$ is a small positive real number.

The details of the method will now be described.
For the sake of simplicity, we assume superadditivity; however, this algorithm could be easily extended to include general cooperative games~\cite{burge2023quantum}.
The goal is to efficiently approximate the Shapley value $\Phi_i$ of a given player $i\in F$.
Suppose quantum representation of function of function $\hat{W}(x)$, defined in Equation~\eqref{eq:What}, is given as:
\begin{equation*}
    U_W \ket{x}\ket{0} := \ket{x} \left( \sqrt{1-\hat{W}(x)}\ket{0} + \sqrt{\hat{W}(x)}\ket{1} \right),
\end{equation*}
where $\ket{x}$ is a vector in the computational basis (i.e., $x\in \{0,1\}^n$).

A critical step for the algorithm is to approximate the weights of the Shapley value function.
These weights correspond perfectly to a slightly modified beta function.
\begin{definition}[Beta function]
We define the beta function as:
\begin{equation*}
\beta_{n,m} = \int\limits_0^1 x^{m}(1-x)^{n-m} dx,\quad 0\leq m \leq n, \quad n, m \in \mathbb{N}.
\end{equation*}
\end{definition}
\begin{lemma}
\label{lem:recurrence}
We have the following recurrence relationship:
\begin{equation*}
\beta_{n,0} = \beta_{n,n} = \frac{1}{n + 1}
\mbox{ and } 
\beta_{n,m} = \frac{m}{n-(m-1)}\beta_{n,m-1}.
\end{equation*}
\end{lemma}
\begin{proof}
There are two cases.

\noindent
Case~1 ($m$ is equal to zero, or $n$). We have the following integration.
\begin{equation*}
\beta_{n,0} = \int\limits_0^1 (1-x)^n dx
= - \frac{(1-x)^{n+1}}{n+1}\bigg\rvert_0^1
= \frac{1}{n+1}
\end{equation*}
A nearly identical calculation can be used to show $\beta_{n,n}$ is equal to $\frac{1}{n+1}$.

\noindent
Case~2 ($0<m<n$). We have the following partial integration.
 \begin{align*}
\beta_{n, m} &= \int\limits_0^1 x^{m}(1-x)^{n-m} dx \\
        &= \frac{x^{m}(1-x)^{n-(m-1)}}{n-(m-1)}  \bigg\rvert_0^1 \\
        &- \int\limits_0^1 \frac{-m}{n-(m-1)} x^{m-1} (1-x)^{n-(m-1)}dx\\
        &= 0 + \frac{m}{n-(m-1)} \int\limits_0^1 x^{m-1} (1-x)^{n-(m-1)}dx\\
        &= \frac{m}{n-(m-1)} 
        \beta_{n,m-1}
 \end{align*}
\end{proof}

\begin{theorem}
The beta function $\beta_{n,m}$ is equal to the Shapley weight function $\gamma(n,m)$, with $0\leq m\leq n$ and $m,n\in \mathbb{N}$.
\end{theorem}
\begin{proof}
The proof is by induction on $m$.

\noindent
Base case ($m=0$). 
According to Case~1 of Lemma~\ref{lem:recurrence},
we have that $\beta_{n,0}$ is equal to $\frac{1}{n+1}$, which is equal to $\gamma(n,0)$.

\noindent
Inductive step ($m>0$). 
Suppose $\beta_{n,k}$ is equal to $\gamma(n,k)$, $k\in\mathbb{N}$, we need to show $\beta_{n,k+1}=\gamma(n,k+1)$, $0 \leq k < n$.
According to Case~2 of Lemma~\ref{lem:recurrence},
$\beta_{n,k+1}$ is equal to $\frac{k+1}{n-k}\beta_{n,k}$.
Using the inductive hypothesis, the latter is equivalent to
$$
\frac{k+1}{n-k}\gamma(n,k) = \frac{k+1}{n-k} \cdot \frac{k!(n-k)!}{(n+1)!}
$$
which matches the definition of $\gamma(n, k+1)$.
\end{proof}
The beta function, and by extension the Shapley weights, 
are approximated with Riemann sums which roughly represent the area under the function $x^m(1-x)^{n-m}$ using partitions of the interval $[0,1]$ with respect to $x$.
It will become apparent that function $x^m(1-x)^{n-m}$ can be implemented efficiently on a quantum computer.
\begin{definition}[Riemann sum]
    A \emph{Riemann sum \cite[page~276]{ross2013elementary}} of a function $f$ with respect to a partition $P=(t_0,\cdots,t_s)$ of the interval $[a,b]$ is an approximation of the integral of $f$ from $a$ to $b$ of the form:
    \begin{equation*}
        \sum\limits_{k=0}^{s-1} (t_{k+1} - t_k) \cdot f(x_k) 
    \end{equation*}
    Where $t_{k+1}-t_k$ is the width of the subinterval, and $f(x_k), x\in[t_k,t_{k+1}]$, is the height.
\end{definition}

The following is the initial quantum state:
\begin{equation}
    \ket{\psi_0} = \ket{0}_\text{Pt} \otimes \ket{0}_\text{Pl} \otimes \ket{0}_\text{Ut},
\end{equation}
where Pt, Pl, and Ut respectively denote the partition, player, and utility registers.
Pt is used it to prepare the $\gamma(n,m)$ weights.
Suppose the number of qubits $\ell\in \mathbb{N}$ in Pt is $\ell=\lceil \log_2an \rceil$, for some fixed $a>0$, thus $\ell=\mathcal{O}(\log an)$.
Then the partition register can be loaded with an arbitrary quantum state in $\mathcal{O}(an)$ time \cite{plesch2011quantum}.
Let the number of qubits in the player register be $n$, one qubit per player excluding player $i$.
Let the number of qubits in the utility register be one.

\medskip

Consider the following function,
\begin{equation*}
    t_\ell(k)=\sin^2\left( \frac{\pi k}{2^{\ell + 1}} \right) \mbox{ with } k=0,1,\ldots,2^{\ell}
\end{equation*}
with which the partition $P_\ell=\left(t_\ell(k)\right)_{k=0}^{2^\ell}$ of interval $[0,1]$ is constructed.
This partition is computationally simple to implement on a quantum computer.
Finally, let us define $w_\ell(k)$ to be the width of the $k^\text{th}$ subinterval of $P_\ell$, $w_\ell(k) = t_\ell(k+1)-t_\ell(k)$,
$k=0,1,\ldots,2^{\ell}-1$.

Let us prepare the partition register to be,
\begin{equation*}
    \sum\limits_{k=0}^{2^\ell-1} \sqrt{w_\ell(k)} \ket{k}.
\end{equation*}
the new state of the quantum system becomes,
\begin{equation*}
    \ket{\psi_1} = \sum\limits_{k=0}^{2^\ell-1} \sqrt{w_\ell(k)} \ket{k}_\text{Pt} \ket{0}_\text{Pl} \ket{0}_\text{Ut}.
\end{equation*}

\begin{figure}
    \centering
    \includegraphics[width=\columnwidth]{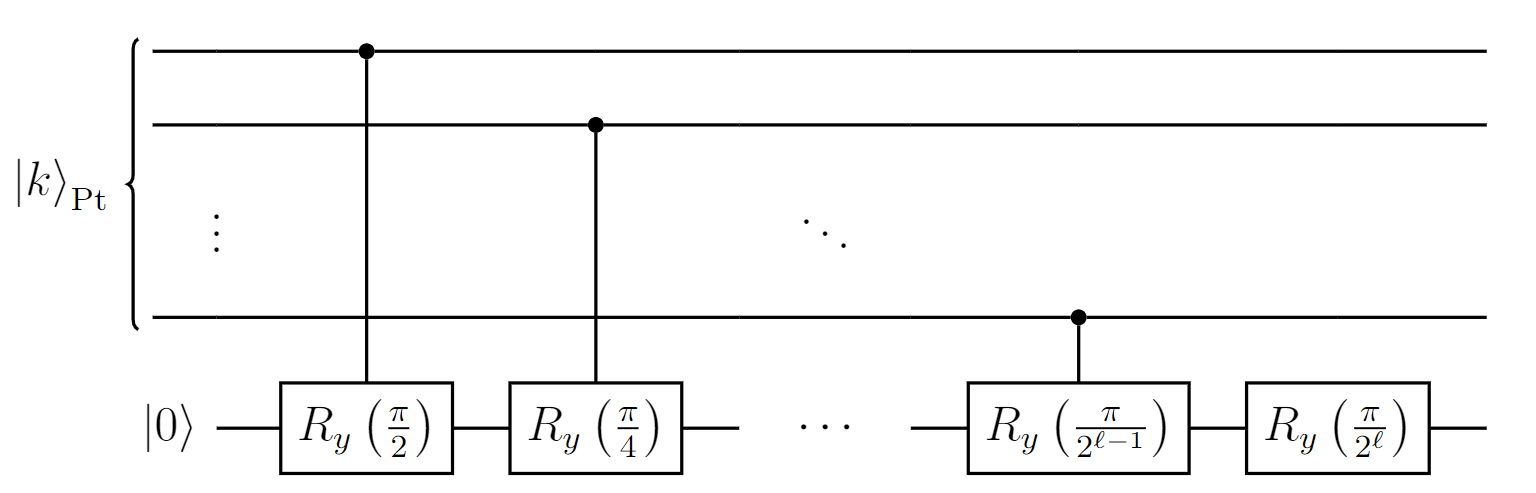}
    \caption{This circuit $R$ is a controlled rotation, where $R_y(\theta) = (\cos(\theta/2), -\sin(\theta/2);\sin(\theta/2), \cos(\theta/2))$.
    (Note: Library used for visualizing circuits can be found here: \cite{kay_2023})
    }
    \label{fig:controlledRotationCircuit}
\end{figure}

Next, perform a series of controlled rotations $R$ (circuit in Figure~\ref{fig:controlledRotationCircuit}) of the form,
\begin{equation*}
    R\ket{k}\ket{0} := \ket{k} \left(\sqrt{1-t'_\ell(k)}\ket{0} + \sqrt{t'_\ell(k)}\ket{1} \right),
\end{equation*}
where $t'_\ell(k) = t_{\ell+1}(2k+1)$ is used to sample the height of the $k^\text{th}$ subinterval in the Riemann sum. 
Note that the value of $t'_\ell(k)$ is always somewhere in the range $[t_\ell(k),t_\ell(k+1)]$.

Unitary $R$ is performed on each qubit in the player register, controlled by the partition register.
The entirety of the applications of $R$ can be performed with $\mathcal{O}(an\log an)$ gates in $\mathcal{O}(an)$ layers and yields the quantum state:

\begin{scriptsize}
\begin{equation*}
    \ket{\psi_2} = \sum\limits_{k=0}^{2^\ell-1} \sqrt{w_\ell(k)} \ket{k}_\text{Pt} \left( \sqrt{1-t_\ell'(k)}\ket{0} + \sqrt{t_\ell'(k)}\ket{1} \right)^{\otimes n} \ket{0}_\text{Ut}.
\end{equation*}
\end{scriptsize}

Note that the player register is of size $n$ qubits.
Let $H_m$ be the set of binary numbers of hamming distance $m$ from $0$ in $n$ bits, then we can rewrite $\ket{\psi_2}$ as:

\begin{scriptsize}
\begin{equation*}
    \ket{\psi_2} = \sum\limits_{k=0}^{2^\ell-1} \sqrt{w_\ell(k)} \ket{k}_\text{Pt} \cdot \sum\limits_{m=0}^n \sqrt{t_\ell'(k)^m (1-t_\ell'(k))^{n-m}} \cdot \sum\limits_{h\in H_m}\ket{h}_\text{Pl} \ket{0}_\text{Ut}
\end{equation*}
\end{scriptsize}

Note that, with this encoding style for $S_x$, $x$'s hamming distance from $0$ in $n$ bits is equal to the size of $S_x$.
In other words, if $h\in H_m$, then $S_h$ contains $m$ players.

\begin{example}
\label{ex:rotationexample}
    Let us consider an example where the number of players is three ($n$ is two).
    We have,
    \begin{align*}
        \Big( \sqrt{1-t_\ell'(k)}&\ket{0} + \sqrt{t_\ell'(k)}\ket{1} \Big) ^{\otimes 2} = 
        \\ 
        &\quad\sqrt{(1-t_\ell'(k))^2}\ket{00} +
        \sqrt{t_\ell'(k)(1-t_\ell'(k))}\ket{01} 
        \\
        &\quad+\sqrt{t_\ell'(k)(1-t_\ell'(k))}\ket{10} + 
        \sqrt{t_\ell'(k)^2}\ket{11}.
    \end{align*}

\begin{figure*}
    \centering
    \includegraphics[scale=0.7]{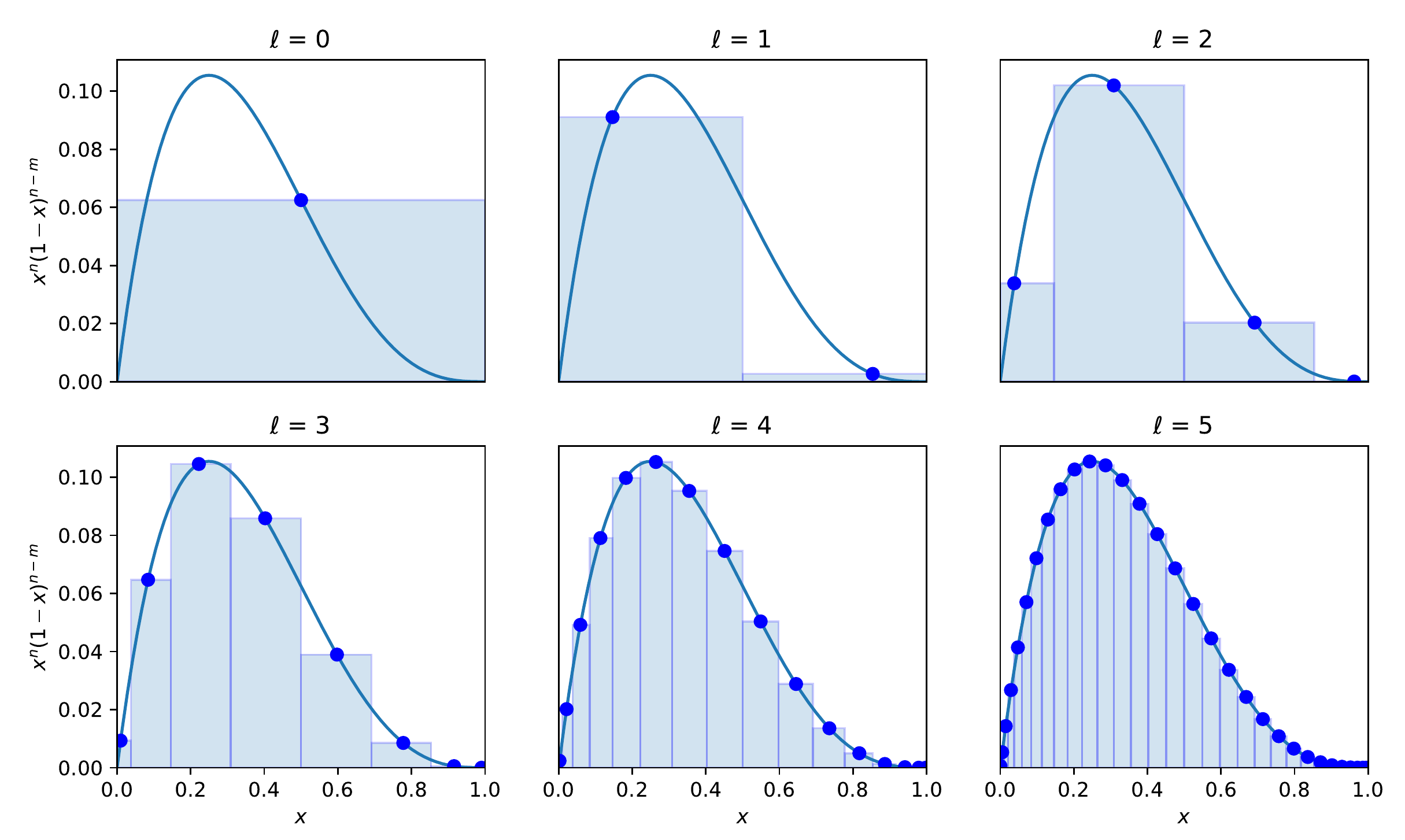}
    \caption{Visual representation of $\beta_{n,m}$ being approximated using Riemann sums with partition $P_\ell$ over function $x^m(1-x)^{n-m}$, $t \in [0,1]$, $n=4$, $m=1$.
    The $k^\text{th}$ rectangle's height is $(t'_\ell(k))^m(1-~t'_\ell(k))^{n-m}$, and its width is $w_\ell(k)$. 
    }
    \label{fig:binomialAreaApprox}
\end{figure*}

Note that $\ket{00}$ is hamming distance $0$ from $\ket{00}$, $\ket{01}$ and $\ket{10}$ are hamming distance $1$ from $\ket{00}$, and $\ket{11}$ is hamming distance $2$ from $\ket{00}$. With this knowledge in hand, we can rewrite the quantum state of Example~\ref{ex:rotationexample} as,
    
\begin{scriptsize}
    \begin{equation*}
        \sqrt{(1-t_\ell'(k))^2}\sum\limits_{h\in H_0}\ket{h} + \sqrt{t_\ell'(k)(1-t_\ell'(k))}\sum\limits_{h\in H_1}\ket{h} + \sqrt{t_\ell'(k)^2}\sum\limits_{h\in H_2}\ket{h}
    \end{equation*}
\end{scriptsize}

This can now be arranged into the  general form,
    \begin{equation*}
    \sum\limits_{m=0}^n \sqrt{t_\ell'(k)^m(1-t_\ell'(k))^{n-m}}\sum\limits_{h\in H_m}\ket{h}, \quad n=2.
    \end{equation*}
This completes the example.
\end{example}

Rearranging the quantum state $\ket{\psi_2}$ gives,

\begin{scriptsize}
\begin{equation*}
    \ket{\psi_2} = \sum\limits_{m=0}^n \sum\limits_{h\in H_m} \sum\limits_{k=0}^{2^\ell-1} \sqrt{w_\ell(k) t_\ell'(k)^m (1-t_\ell'(k))^{n-m}} \ket{k}_\text{Pt} \ket{h}_\text{Pl} \ket{0}_\text{Ut}.
\end{equation*}
\end{scriptsize}

Next, we perform $U_W$ on the utility register controlled by the player register.
For convenience, let us for the moment write $U_W \ket{h}\ket{0} = \ket{h}\ket{W(h)}$, where,
\begin{equation*}
    \ket{W(h)} = \sqrt{1-\hat{W}(h)} \ket{0} + \sqrt{\hat{W}(h)} \ket{1}
\end{equation*}

Applying $U_W\ket{h}_\text{Pl}\ket{0}_\text{Ut}$ gives $\ket{\psi_3}$, which is equal to,

\begin{scriptsize}
\begin{align*}
     \sum\limits_{m=0}^n \sum\limits_{h\in H_m} \sum\limits_{k=0}^{2^\ell-1} \sqrt{w_\ell(k) t_\ell'(k)^m (1-t_\ell'(k))^{n-m}} \ket{k}_\text{Pt} \ket{h}_\text{Pl} \ket{W(h)}_\text{Ut}.
\end{align*}
\end{scriptsize}

This operation is wholly dependent on the game being analyzed and its complexity.
Assuming the algorithm is being implemented with a look-up table, one could likely use qRAM \cite{giovannetti2008quantum}.
This approach has a time complexity of $\mathcal{O}(n)$ at the cost of $\mathcal{O}(2^n)$ qubits for storage.
However, depending on the problem, there are often far less resource-intense methods of implementing $U_W$, 
as will be seen with the implementation of weighted voting games (Section~\ref{sec:example}).

This is the final quantum state.
Let us now analyze this state through the lens of density matrices.
Taking the partial trace with respect to the partition and player registers yields,

\begin{scriptsize}
\begin{align*}
    \tr_{\text{Pt},\text{Pl}}\left(\ketbra{\psi_3}{\psi_3}\right)=
    \sum\limits_{m=0}^n \sum\limits_{h\in H_m} & \left(\sum\limits_{k=0}^{2^\ell-1} w_\ell(k) t_\ell'(k)^m (1-t_\ell'(k))^{n-m} \right)\\
    & \cdot \ket{W(h)}_\text{Ut} \bra{W(h)}_\text{Ut}.
\end{align*}
\end{scriptsize}

\begin{theorem}
    The Riemann sum using partition $P_\ell$ to approximate area under $x^m(1-x)^{n-m}$ for $x\in[0,1]$ asymptotically approaches $\gamma(n,m)$.
    Formally,
    \begin{equation} \label{eq:riemannSumApproachesGamma}
        \lim\limits_{\ell \rightarrow \infty} \sum\limits_{k=0}^{2^\ell-1} w_\ell(k) t_\ell'(k)^m (1-t_\ell'(k))^{n-m} = \gamma(n,m)
    \end{equation}
\end{theorem}
\begin{proof}
Let $f(x)=x^m(1-x)^{n-m}$, and recall our definition for partition of $[0,1]$, $P_\ell=\left(t_\ell(k)\right)_{k=0}^{2^\ell}$.

Define $\text{mesh}(P_\ell)=\sup\{w_\ell(k): k=0,1,\dots,2^\ell -1\}$ \cite[page~275]{ross2013elementary}.
Since,
\begin{equation*}
    w_\ell(k) = \sin^2\left( \frac{\pi (k+1)}{2^{\ell+1}} \right) - \sin^2\left( \frac{\pi k}{2^{\ell+1}} \right),
\end{equation*}
we can bound $w_\ell(k)$,
\begin{equation*}
    w_\ell(k) \leq  \frac{\pi}{2} \left(\left( \frac{\pi (k+1)}{2^{\ell+1}} \right)- \left( \frac{\pi k}{2^{\ell+1}} \right) \right) = \frac{\pi^2}{2^{\ell+2}}.
\end{equation*}
This is valid because $d/dx \sin^2((\pi/2)x) \leq d/dx (\pi/2)x$ for all $x\in[0,1]$.
It is then clear that as $\ell\rightarrow \infty$, $\text{mesh}(P_\ell)\rightarrow0$.
Additionally, $f$ is integrable over $[a,b]$, as it is a polynomial and hence continuous \cite[page~278]{ross2013elementary}.

By \cite[page~278,~Corollary~32.10]{ross2013elementary}, as $\ell$ increases, any Riemann sum of $f$ using partition $P_\ell$ converges to:
\begin{equation*}
    \int_0^1 x^m(1-x)^{n-m},
\end{equation*}
which is equal to $\gamma(n,m)$.
\end{proof}

The Riemann sum approximation of the modified beta function, $\beta_{n,m}=\gamma(n,m)$, is visualized in Figure~\ref{fig:binomialAreaApprox}.
Applying a Riemann sum approximation for $\gamma(n,m)$, we have,
\begin{equation*}
    \tr_{\text{Pt},\text{Pl}}(\ketbra{\psi_3}{\psi_3})\approx\sum\limits_{m=0}^n \sum\limits_{h\in H_m} \gamma(n,m) \ket{W(h)}_\text{Ut} \bra{W(h)}_\text{Ut}.
\end{equation*}

Finally, suppose we measuring the utility register in the computational basis.
This yields the following expected value with empirical error less than $\mathcal{O}(a^{-1})$ (Section~\ref{sec:example}):
\begin{equation*}
    \sum\limits_{m=0}^n \sum\limits_{h\in H_m} \gamma(n,m) \hat{W}(h)
\end{equation*}

Plugging in the definition for $\hat{W}$, we have
\begin{equation*}
    \frac{1}{W_\text{max}} \sum\limits_{m=0}^n \sum\limits_{h\in H_m} \gamma(n,m) \left(V\left(S_h \cup \{i\}\right) - V\left(S_h\right)\right)
\end{equation*}

Notice that, in the $S_x$ encoding, $H_m$ represents each subset of $F\setminus\{i\}$ of size $m$.
As a result, the equation is, in effect, summing over each subset of $F\setminus\{i\}$.
As a final step, multiply by $W_\text{max}$:
\begin{equation*}
    \sum\limits_{S\subseteq F\setminus\{i\}} \gamma(\lvert F \setminus \{i\} \rvert, \lvert S \rvert) \cdot \left(V\left(S \cup \{i\}\right) - V\left(S\right)\right).
\end{equation*}
This expected value is precisely the Shapley value $\Phi_i$ to some error which is shown empirically in Section~\ref{sec:example}. 

With the ability to craft these states, we can now extract the required information to find a close approximation to the Shapley value.
Assuming we could get the expected value instantly, we would (empirically; see Section~\ref{sec:example}) have an error of $\mathcal{O}(a^{-1})$.
However, it takes some work to approximate an expected value.
This can be achieved with accuracy $\epsilon$ with some chosen confidence using amplitude amplification~\cite{montanaro2015quantum}, resulting in circuit depth $\mathcal{O}(\epsilon^{-1})$ times our algorithm.
Alternatively, for simplicity, we can repeatedly construct state $\ket{\psi_3}$ and measure the utility register roughly $\mathcal{O}(\epsilon^{-2})$ times.
The resulting measurements of $\ket{0}$s and $\ket{1}$s can be analyzed using binomial distributions.
We can estimate the underlying probability, from which we can extract the Shapley value, within some confidence interval \cite{wallis2013binomial}.
This estimation takes up the bulk of the runtime of the algorithm.
Thus the total error can be as low as $\mathcal{O}(W_{\max}\cdot(a^{-1}+\epsilon))$.
This can be even further modified to depend on the standard deviation of $V$, based on the techniques from \cite{montanaro2015quantum}.

\section{Numerical Examples}
\label{sec:example}

Perhaps we can help David solve his problem (cf. Section~\ref{sec:WeightedVotingGame}) using our quantum approach. 
We intend to apply the method presented in Section~\ref{sec:algorithm} for weighted voting games. Additional results, together with the simulation code, are available in a \href{https://github.com/iain-burge/QuantumShapleyValueAlgorithm}{companion github repository} \cite{githubEntry}.
Let us approximate each player's Shapley value, $\Phi_i$.
We have a game $G=(F,V)$, where $F=\{0,1,2\}$, $n=2$, and $V$ is defined in Equation~\ref{eq:voteGameV}.
In this case, $W(S)$, $S\subseteq F\setminus \{i\}$,
represents whether or not player $i$ has the deciding vote assuming those in $S$ are voting for Chinese food.
For example, Alice (0) is a deciding vote for $S=\{1\}$ ($S$ contains Bob).
Let the voting weights be $w_0=3$, $w_1=2$, and $w_2=1$.
Thus, we can define $W(S)$, $S\in F\setminus \{i\}$, for player $i$ to be,

\begin{equation*}
    W(S)=W'\left(\sum\limits_{j\in S} w_j\right)
\end{equation*}
where,
\begin{equation*}
    W'(Votes)= 
    \begin{cases}
        1 &\quad \text{if } q-w_i \leq Votes < q, \\
        0 &\quad \text{otherwise.}
    \end{cases}
\end{equation*}

Note that, $W(S)$ is either $0$ or $1$.
Thus, we can take $W_\text{max}=1$, so $\hat{W}(x) = W(S_x)$.
Thus we can define $U_W$ to be:
\begin{equation*}
    U_W\ket{x}\ket{0} = \ket{x} \otimes \left[ \left(1-\hat{W}(x)\right)\ket{0} + \hat{W}(x)\ket{1}\right]
\end{equation*}
The quantum circuit for the scenario is shown in Figure~\ref{fig:U_wOfVotingGame}.

\begin{figure}
    \centering
    \includegraphics[width=\columnwidth]{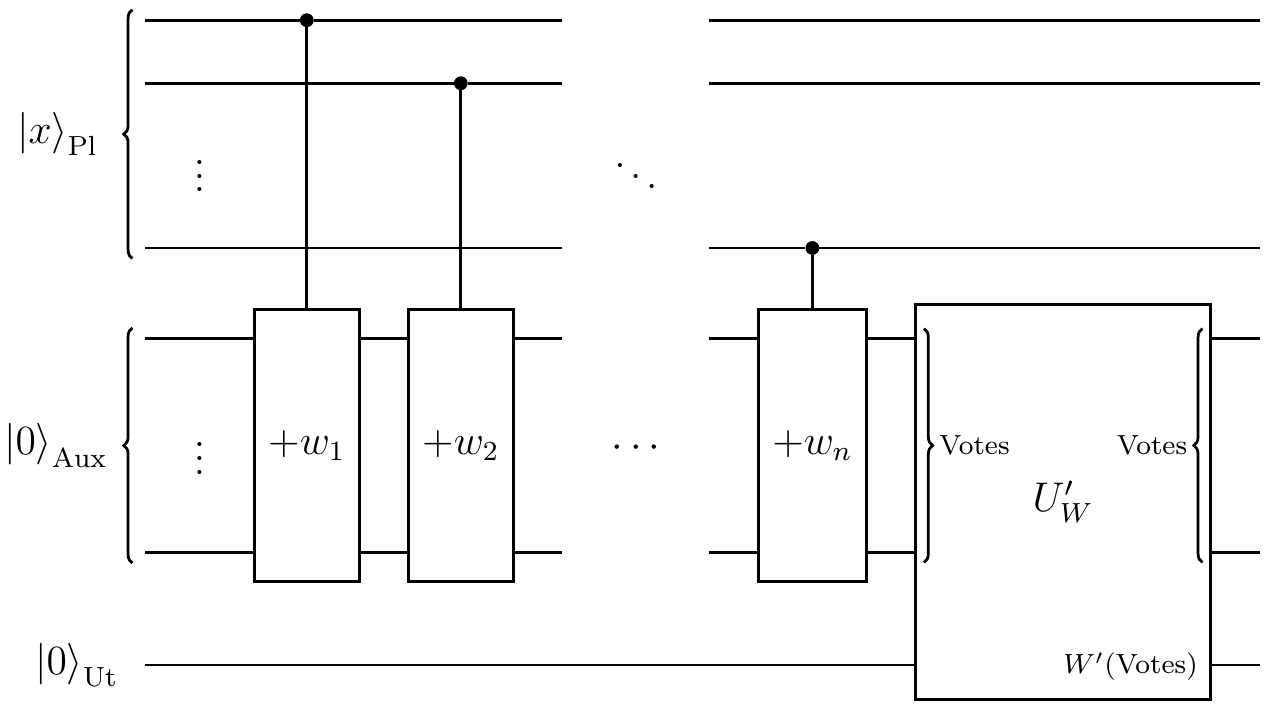}
    \caption{Circuit of $U_W$ for a weighted voting game.
    This circuit takes an input $x$ and outputs $W(S_x)$ in the utility register (Recall, $S_x$ is defined in Section~\ref{sec:algorithm}).
    The auxiliary register contains the total vote count.
    Just before $U_W'$, the Aux register is in a basis state corresponding to the vote count of $S_x$.
    $U_W'$ uses the auxiliary register as an input, and outputs whether the vote count is in the correct range for player $i$ to be a deciding vote. 
    Note that there is no $+w_i$ gate activation, as $S_x$ does not include player $i$ by definition.
    Results and simulation code are available in a \href{https://github.com/iain-burge/QuantumShapleyValueAlgorithm}{companion GitHub repository} \cite{githubEntry}.
    }
    \label{fig:U_wOfVotingGame}
\end{figure}

With $U_W$ defined, all other steps are entirely agnostic to voting games.
Let $\ell=2$, and suppose for simplicity's sake ${\epsilon=0}$.
This is not a realistic scenario, but it demonstrates how quickly the expected value of the utility register converges.
With these parameters, we get the following approximations for the Shapley values:
\begin{equation*}
    \tilde\Phi_0 \approx 0.6617,\quad \tilde\Phi_1, \tilde\Phi_2 \approx 0.1616 
\end{equation*}
The direct calculation for Alice can be seen in the Appendix.

To rigorously demonstrate efficacy, we performed many trials on random weighted voting games (Figure~\ref{fig:expectedValueError}).
Visual inspection confirms that, in most cases, when $\ell$ is incremented, the error is divided by more than a factor of two (so the reciprocal more than doubles).
As incrementing $\ell$ implies doubling the work to prepare $\ket{\psi_1}$, we have empirically shown a linear or better relationship between complexity and error for Step 1 (Section~\ref{sec:algorithm}).
So, the empirical error for this step is $\mathcal{O}(a^{-1})$.

Step 2 depends entirely on the game. However, if it is possible to implement on a classical computer, much like with Grover's algorithm, then it can be implemented in a quantum setting~\cite{grover1996fast}.

\begin{figure*}
    \centering
    \includegraphics[scale=0.7]{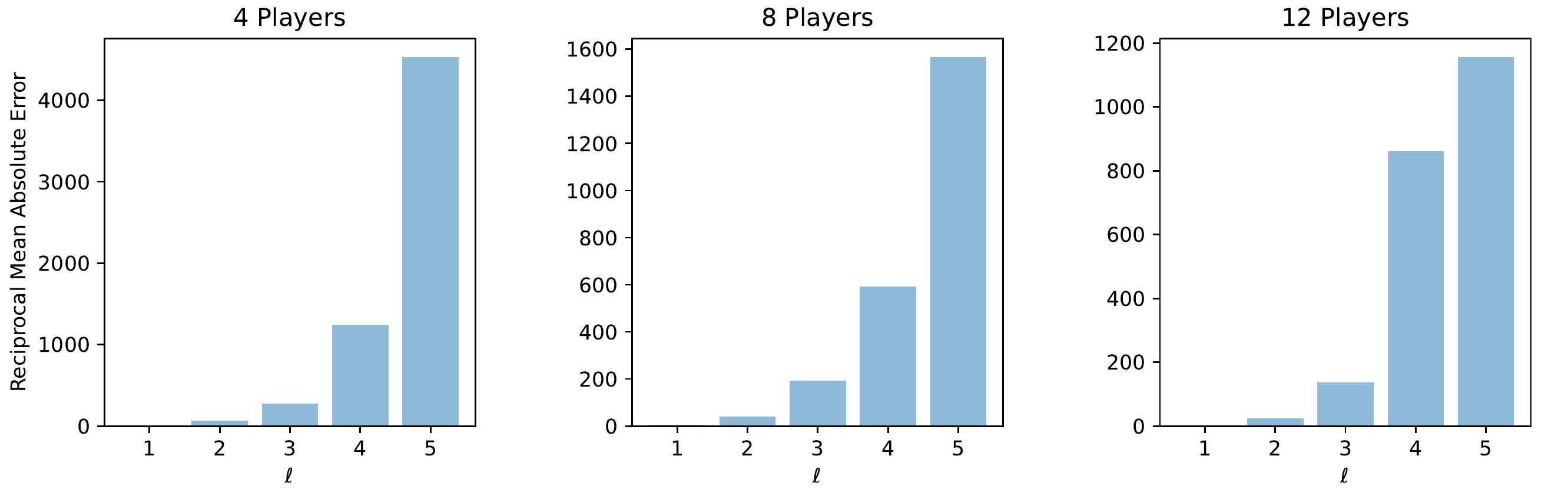}
    \caption{
        We generated 32 random weighted voting games for each condition.
        We generated random weights $w_j\in \mathbb{N}$ for each case such that $q\leq \sum_j w_j < 2q$.
        There were three primary scenarios: 
        (1) Four players, voting threshold $q=8$; 
        (2) Eight players, voting threshold $q=16$; and
        (3) 12 players, voting threshold $q=32$.
        We approximated every player's Shapley value for each scenario with our quantum algorithm using various $\ell$'s, assuming $\epsilon=0$.
        Next, we found the absolute error of our approximation by comparing each approximated Shapley value to its true value.
        Finally, we took the reciprocal of the mean for all Shapley value errors in each random game for each condition.
    }
    \label{fig:expectedValueError}
    
\end{figure*}

Step 3 is well studied, and there are two valid approaches to finding the expected value of measuring the utility register with error $\epsilon$.
Either (i), repeat steps one and two of Section~\ref{sec:algorithm} and measure the utility register $\mathcal{O}(\epsilon^{-2})$ times.
Otherwise, proceed with (ii), use amplitude estimation as described in Section~\ref{sec:algorithm} with the techniques described in \cite{montanaro2015quantum}.
Note that (ii) results in an increase to the circuit depth by $\mathcal{O}(\epsilon^{-1})$ times, but only needs to be repeated $\mathcal{O}(1)$ times assuming a fixed confidence probability.

Let us summarize based on our empirical findings and the assumption $W$ is more complex than Step 1. 
We have that complexity in terms of applications of $W$ is $\mathcal{O}(\epsilon^{-1})$ in the case of weighted voting games.
For more complex problems, where the value function $W$ is not restricted to $W:\mathcal{P}(F) \rightarrow [0,1]$ the approach can be slightly modified such that $W$ is applied $\mathcal{O}(\sigma\epsilon^{-1})$ times \cite{montanaro2015quantum}, where $\sigma$ is the standard deviation of the $W$ given the $\gamma$ distribution.

\section{Conclusion}
\label{sec:conclusion}

We have addressed one of the main problems of generating post hoc explanations of \gls*{ml} models in the quantum context. 
More precisely, we have addressed the problem of efficiently generating post hoc explanations using the concept of the Shapley value. 
Under the \gls*{xqml} context, the challenge of explainability is amplified since measuring a quantum system destroys the information. 
Using the classical concept of Shapley values for post hoc explanations in quantum computing does not translate trivially.  
We have proposed novel algorithms to reduce the problem of accurately estimating the Shapley values of a quantum algorithm into the solved problem of estimating the expected value of a quantum algorithm. 
Finally, we have determined the efficacy of the algorithms by using empirical and simulation methods.

\medskip

\noindent \textbf{Acknowledgments ---}
We thank Prof. Yuly Billig (Carleton University) for his notational suggestion regarding hamming distances.
We would also like to thank a peer of Iain, Michael Ripa, for many productive conversations regarding this research.
The research was supported by \gls*{nserc} of Canada.

\appendix

Quantum estimation of Alice's Shapley value by hand.
Let $\ell=2$.
Note that the Auxiliary register stores vote count.
We begin with the state:
\begin{equation*}
    \ket{00}_\text{Pt}\ket{00}_\text{Pl}\ket{000}_\text{Aux}\ket{0}_\text{Ut}
\end{equation*}

We perform the first step of the algorithm from Section \ref{sec:algorithm}, yielding:

\begin{footnotesize}
\begin{align*}
    \sum\limits_{k=0}^{3} &w_2(k) \ket{k}_\text{Pt} \Big[(1-t_2'(k))\ket{00}_\text{Pl}+\sqrt{t_2'(k)(1-t_2'(k))}\ket{01}_\text{Pl}\\
    &+\sqrt{t_2'(k)(1-t_2'(k))} \ket{10}_\text{Pl}+t_2'(k)\ket{11}_\text{Pl}\Big] \ket{000}_\text{Aux}\ket{0}_\text{Ut}
\end{align*}
\end{footnotesize}

Next, we tally the votes.
This step is the first half of the circuit in Figure \ref{fig:U_wOfVotingGame}, up to and without including $U_W'$.
We get,

\begin{footnotesize}
\begin{align*}
    \sum\limits_{k=0}^{3} &w_2(k) \ket{k}_\text{Pt} 
    \Big[
    (1-t_2'(k))\ket{00}_\text{Pl}\ket{000}_\text{Aux}\\&
    +\sqrt{t_2'(k)(1-t_2'(k))}\ket{01}_\text{Pl}\ket{001}_\text{Aux}\\&
    +\sqrt{t_2'(k)(1-t_2'(k))} \ket{10}_\text{Pl}\ket{010}_\text{Aux}\\&
    +t_2'(k)\ket{11}_\text{Pl}\ket{011}_\text{Aux}
    \Big] 
    \ket{0}_\text{Ut}
\end{align*}
\end{footnotesize}

Performing the remainder of $U_W$ gives,
\begin{footnotesize}
\begin{align*}
    \sum\limits_{k=0}^{3} &w_2(k) \ket{k}_\text{Pt} 
    \Big[
    (1-t_2'(k))\ket{00}_\text{Pl}\ket{000}_\text{Aux}\ket{0}_\text{Ut}\\&
    +\sqrt{t_2'(k)(1-t_2'(k))}\ket{01}_\text{Pl}\ket{001}_\text{Aux}\ket{1}_\text{Ut}\\&
    +\sqrt{t_2'(k)(1-t_2'(k))} \ket{10}_\text{Pl}\ket{010}_\text{Aux}\ket{1}_\text{Ut}\\&
    +t_2'(k)\ket{11}_\text{Pl}\ket{011}_\text{Aux}\ket{1}_\text{Ut}
    \Big] 
\end{align*}
\end{footnotesize}

The expected value when measuring the utility register is $\approx0.6617$, a close approximation for $\Phi_0 = 2/3 = 0.6666\cdots$.

\balance 

\end{document}